\documentclass[12pt]{article}
\usepackage{amsmath}
\usepackage{amsfonts}
\usepackage{amssymb}
\usepackage{amsthm}
\usepackage{stackrel}

\usepackage{color}

\newtheorem{thm}{Theorem}[section]

\newcommand{\R}{{\rm I}\kern-0.18em{\rm R}}
\newcommand{\1}{{\rm 1}\kern-0.25em{\rm I}}
\newcommand{\E}{{\rm I}\kern-0.18em{\rm E}}
\newcommand{\p}{{\rm I}\kern-0.18em{\rm P}}
\makeatletter

\usepackage{epsfig}

\usepackage{fancyhdr}

\usepackage{graphics}

\usepackage{amssymb}
\usepackage{amsxtra}
\usepackage{amsfonts}

\usepackage{afterpage}
\usepackage{eucal}
\usepackage{eufrak}

\usepackage{enumerate}

\allowbreak

\author{Lev B Klebanov\footnote{Department of Probability and Statistics, MFF, Charles University, Prague-8, 18675, Czech Republic, e--mail: levbkl@gmail.com}, Gregory Temnov\footnote{Department of Probability and Statistics, MFF, Charles University, Prague-8, 18675, Czech Republic, e--mail: levbkl@gmail.com}, \\Ashot V. Kakosyan\footnote{Yerevan State University}}
\title{Some Contra-Arguments for the Use of Stable Distributions in Financial Modeling}
\date{}

\begin{document}
\maketitle

\begin{abstract}

In the present paper, we discuss contra-arguments concerning the use of Pareto-Lev\'{y} distributions for modeling in Finance. It appears that such probability laws do not provide sufficient number of outliers observed in real data. Connection with the classical limit theorem for heavy-tailed distributions with such type of models is also questionable. The idea of alternative modeling is given. 

\noindent
{\bf keywords}:outliers; financial indexes; heavy tails; stable distributions

\end{abstract}

\section{Introduction}\label{s1}
\setcounter{equation}{0}

During latest decades, we have seen large numbers of publications on the use of Pareto-Lev\'{y} and other heavy-tailed distributions in Finance (see, for example, \cite{RM}, \cite{HB} and references there). One of the key authors that initiated  this, was Mandelbrot. He mentioned that Gaussian distribution cannot provide a solid explanation for observed large amount of `outliers' - in other words, the number of observations for which absolute value of deviation from empirical mean is bigger than $k s$, where $s^2$ is empirical variance. Aparently, Mandelbrot considered stable distributions as a unique alternative to Gaussian family basing on the fact of classical limit theorem for independent identically distributed (i.i.d.) random variables, or, equivalently, on stability property. In the remainder of this paper, we show that stable distributions cannot explain such a large number of outliers either. The connection to classical limit theorem for a non-random number of random summands seems to be questionable. A more natural perspective comes from the use of limit theorem for a random number of random variables.

\section{Probability of outliers for large samples}\label{s2} 
\setcounter{equation}{0}

In this Section we look at the probability of observing outliers in the case of distributions belonging to a domain of attraction of strictly stable distribution \footnote{For the definition and description of the domain of attraction of stable distributions see \cite{IL}}.

Specifically, suppose that $X_1,X_2, \ldots ,X_n$ is a sequence of i.i.d. random variables. Denote by
\[ \bar{x}_n=\frac{1}{n}\sum_{j=1}^{n}X_j, \;\; s^2_n=\frac{1}{n}\sum_{j=1}^{n}(X_j-\bar{x})^2\]
their empirical mean and empirical variance correspondingly. Let $k>0$ be a fixed number. We are interested in the following probability
\begin{equation}
\label{eq1}
p_n = \p\{|X_1-\bar{x}_n|>k s_n \}.
\end{equation}
Our aim here is to prove the following Theorem.
\begin{thm}
\label{th1} 
Suppose that $X_1,X_2, \ldots ,X_n$ is a sequence of i.i.d. random variables belonging to a domain of attraction of strictly stable random variable with index of stability $\alpha \in (0,2)$. Then
\begin{equation}
\label{eq2}
\lim_{n \to \infty}p_n =0.
\end{equation}
\end{thm}
\begin{proof}
Since $X_j,\; j=1, \ldots ,n$ belong to the domain of attraction of strictly stable random variable with index $\alpha <2$, it is also true that $X_1^2, \ldots , X_n^2$ belong to the domain of attraction of one-sided stable distribution with index $\alpha /2$. 

1) Consider at first the case $1<\alpha <2$. In this case, $\bar{x}_n \stackrel[n \to \infty]{}{\longrightarrow }a=\E X_1$ and
$s_n \stackrel[n \to \infty]{}{\longrightarrow }\infty$. We have
\[ \p\{|X_1 - \bar{x}_n| >k s_n\} = \p\{X_1> ks_n+\bar{x}_n\}+\p\{X_1< -ks_n+\bar{x}_n\} = \]
\[=\p\{X_1> ks_n+a+o(1)\}+\p\{X_1< -ks_n+a +o(1)\}  \stackrel[n \to \infty]{}{\longrightarrow } 0. \]

2) Suppose now that $0<\alpha < 1$. In this case, we have $\bar{x}_n \sim n^{1/\alpha -1}Y$ as $n \to \infty$. Here $Y$ is $\alpha$-stable random variable, and the sign $\sim$ is used for asymptotic equivalence. Similarly, 
\[ s^2_n =\frac{1}{n}\sum_{j=1}^{n}X_j^2 - \bar{x}^2_n \sim n^{2/\alpha -1}Z (1+o(1)), \]
where $Z$ has one-sided positive stable distribution with index $\alpha /2$. We have
\[ \p\{|X_1-\bar{x}_n|>k s_n\}=\p\{ (X_1-\bar{x}_n)^2>k s^2_n\}=\]
\[=\p\{X_1^2 > n^{2/\alpha -1}Z (1+o(1))\}  \stackrel[n \to \infty]{}{\longrightarrow } 0. \] 

3) In the case $\alpha =1$ we deal with Cauchy distribution. The proof for this case is very similar to that in the case 2). We omit the details.
\end{proof}

Let us note that for the case of distributions having finite second moment and non-compact support, the probability (\ref{eq1}) has a positive limit as $n \to \infty$. 
Indeed, if the second moment is finite then both $\bar{x}_n$ and $s_n$ have finite limits $a=\E X$ and $\sigma^2 = \E X^2 - (\E X)^2$. The probability (\ref{eq1}) converges to
$\p\{|X_1-a|> k \sigma\} >0$ as $n \to \infty$ because of non-compactness of support of the distribution. Therefore, it is clear that the explanation of  presence a large number of outliers cannot be provided by the heaviness of the tails. 

\section{Simulation study}\label{s3}
\setcounter{equation}{0}

Although the results of Section \ref{s2} show that asymptotically the probability of presence of a large number of outliers is negligible, we may still have doubts for not too large samples. However, the samples in practice are not too small. So, basing on actual observed data, we can consider the sample size $n$ of order $50,000$ as typical one. 

1. We simulated $m=1500$ samples of size $n$ ($n$ is growing from 1,000 to 25,000 with step 2,000) and calculated the estimate of probability $p_n$ given by (\ref{eq1}) for $k=3$. The behavior of this probability as a function of $n$ is reflected on Figure \ref{fig1}.  Blue line corresponds to symmetric stable distribution with $\alpha=1.2$; red line - to standard Gaussian distribution.
We see that the probability $p_n$ for $\alpha=1.2$-stable symmetric distribution becomes smaller that for Gaussian case starting at about $n=18,000$. Therefore, we cannot expect many outliers for such stable distribution for typical sample size. 

\begin{figure}[h]
	\centering
	\hfil
	\includegraphics[scale=0.8]{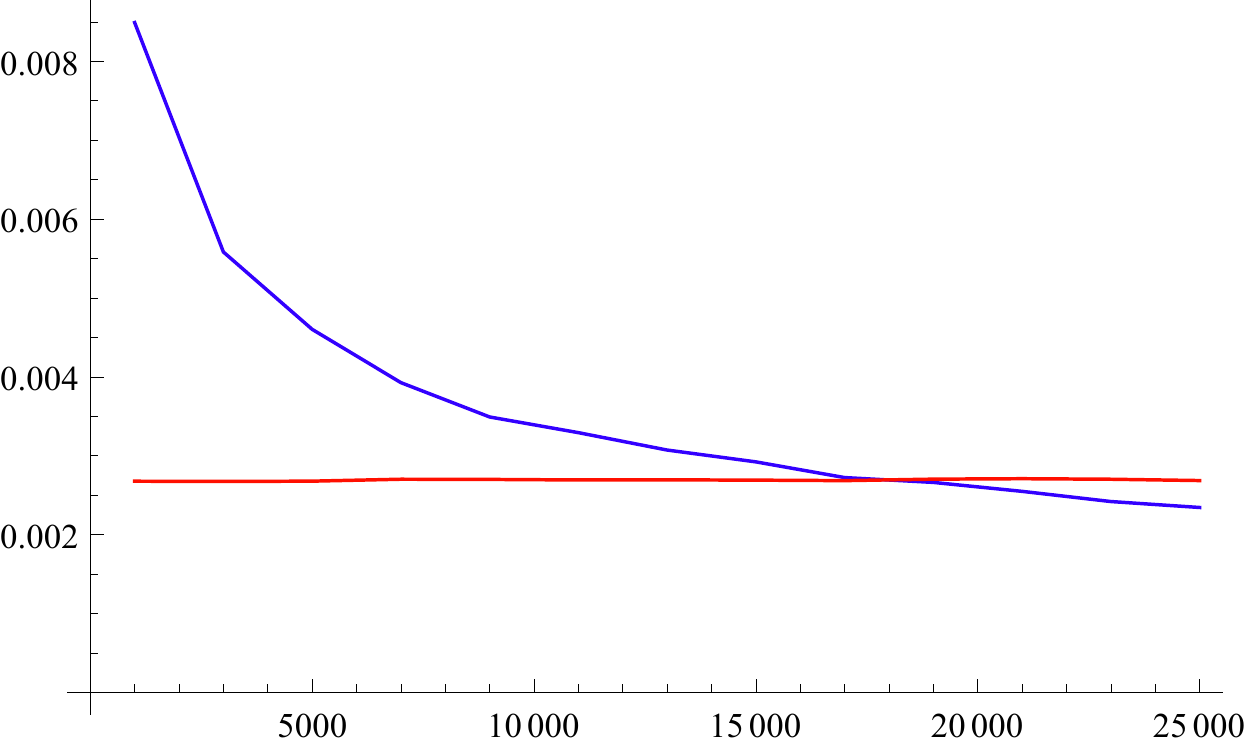}
	\caption{Probability $p_n$ for $k=3$. Blue line corresponds to symmetric stable distribution with $\alpha=1.2$; red  line - to standard Gaussian distribution}\label{fig1} 
 \end{figure} 

2. Again, we simulated $m=1500$ samples of size $n$ ($n$ is again growing from 1,000 to 25,000 with step 2,000) and calculated the estimate of probability $p_n$ given by (\ref{eq1}) for $k=2.5$. The behavior of this probability as a function of $n$ is plotted on Figure \ref{fig2}.  Blue line corresponds to symmetric stable distribution with $\alpha=1.8$; red line - to standard Gaussian distribution. We see that the probability $p_n$ for $\alpha=1.8$-stable symmetric distribution becomes smaller than that for Gaussian case starting from about $n=4,000$. 

The decrease of $p_n$ for the case $k=3$ is much slower. For example, $p_{50000}=0.00591093$ for  $\alpha=1.8$-stable symmetric distribution versus $0.0026998$ for Gaussian distribution.  In this situation, we cannot say that stable distribution provides less outliers that Gaussian law. But will the corresponding number of outliers be sufficient for the predicted figures to be in agreement with observed data? Unfortunately, there are only a few papers giving observed number of outliers, that is an estimate of probability $p_n$ for $k=3$. Some data of this kind may be found in \cite{EK}, Table 3. Corresponding estimates for probabilities $|X_1-\bar{x}_n|>3 s_n$ given there vary from 0.009 to 0.013. For symmetric stable distribution with $\alpha =1.8$, this probability is $p_{50000}=0.00591093$. We see that it is too small to explain the number of outliers in Table 3 from \cite{EK}.

Let us note that similar simulations for different values of parameters and sample size are given in \cite{Kl}. We shall not discuss them here.

\begin{figure}[h]
	\centering
	\hfil
	\includegraphics[scale=0.8]{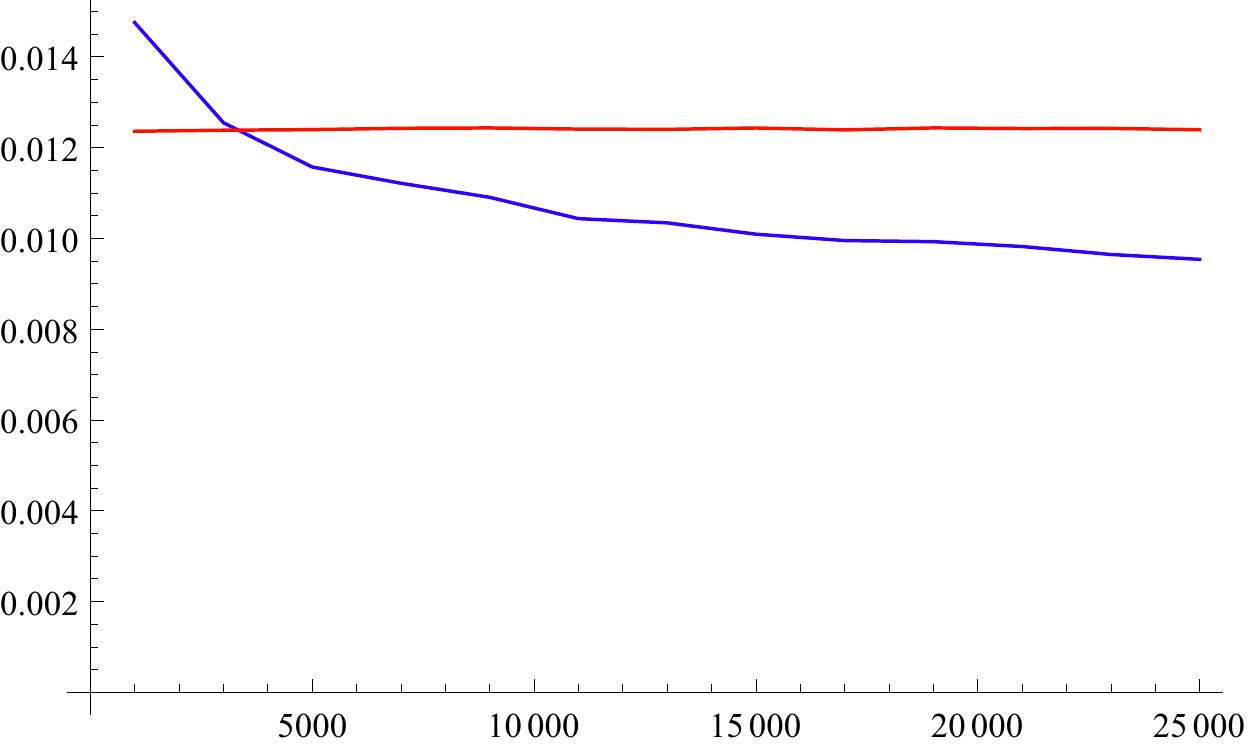}
	\caption{Probability $p_n$ for $k=2.5$. Blue line corresponds to symmetric stable distribution with $\alpha=1.8$; red line - to standard Gaussian distribution}\label{fig2} 
\end{figure} 

\section{Tempered stable distributions}\label{s3a}
\setcounter{equation}{0}

Lately, {\it tempered stable distributions} have been growing increasingly popular. The idea of such distribution lies in altering the tails of stable laws by exponential tails starting from a certain point of distribution's support. Of course, the probability $p_n$ in this case will not converge to zero anymore. However, one may expect that this probability will be small if the point of the tails' alteration is far from the origin, since tempered distribution will be close to stable in this case. Simulations support this opinion. 

Let us consider only the case of symmetric tempered stable distributions having characteristic functions
\begin{equation}
\label{eqA}
f(u,\alpha,\lambda) =\exp\{A\bigl( (\lambda-iu)^{\alpha}+(\lambda+i u)^{\alpha}-2\lambda^{\alpha}\bigr)\}
\end{equation}
for $\alpha \in (1,2)$, $A>0$, $\lambda>0$. 

On Figure \ref{fig2a}, plot of limit probability $p(\alpha,\lambda) =\lim_{n \to \infty}p_n$ is given for the distribution (\ref{eqA}) as a function of parameters $\alpha$ and $\lambda$ for $A=1$. This plot shows that the probabilities are too small to explain the presence of large number of outliers.

\begin{figure}[h]
	\centering
	\hfil
	\includegraphics[scale=0.8]{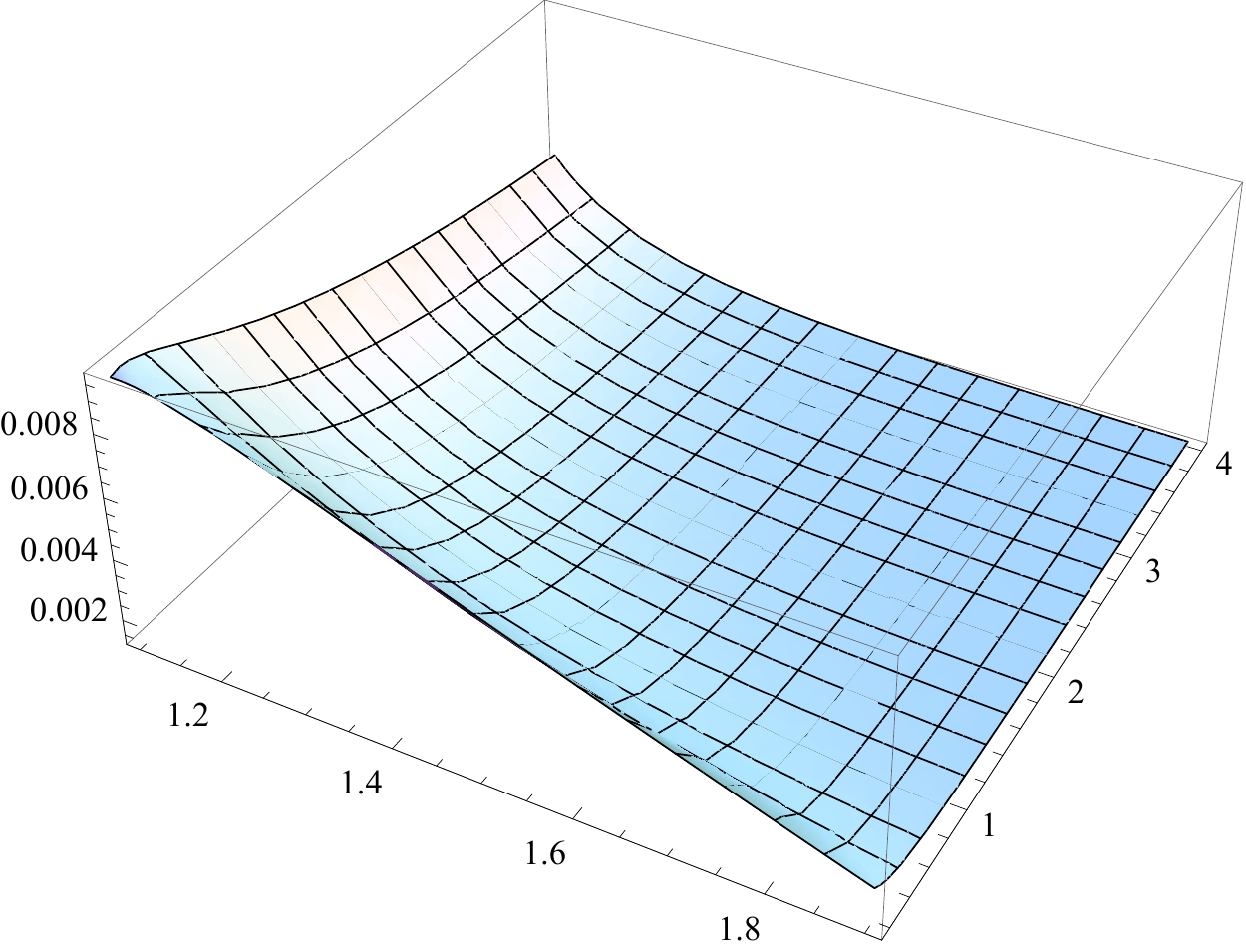}
	\caption{Probability $p_n$ as a function of parameters $\alpha \in (1.1,1.9)$ and $\lambda \in (0.2,4)$.}\label{fig2a} 
\end{figure}

\section{How to obtain more outliers?}\label{s4}
\setcounter{equation}{0}

Here we discuss a way of constructing from a distribution another one having a higher probability to observe outliers. We call this procedure "put tail down".

Let $F(x)$ be a probability distribution function of random variable $X$ having finite second moment $\sigma^2$ and such that $F(-x) = 1-F(x)$ for all $x \in \R^1$. Take a parameter $p \in (0,1)$ and fix it. Define a new function 
\[ F_p(x)=(1-p)F(x) + p H(x),\]
where $H(x) = 0$ for $x<0$, and $H(x)=1$ for $x>0$. It is clear that $F_p(x)$ is probability distribution function for any $p \in (0,1)$. Of course, $F_p$ also has finite second moment $\sigma_p^2$, and $F_p(-x)=1-F_p(x)$. However, $\sigma_p^2=(1-p)\sigma^2 ,\sigma^2$. Let $Y_p$ be a random variable with probability distribution function $F_p$. Then
\[ \p\{|Y_p|>k\sqrt{1-p}\sigma\} = 2 \p\{Y_p>k\sqrt{1-p}\sigma\} =2(1-p)\bigl(1-F(k\sqrt{1-p}\sigma)\bigr).\]
Denoting $\bar{F}(x)=1-F(x)$ rewrite previous equality in the form
\begin{equation}
\label{eq3}
\p\{|Y_p|>k\sqrt{1-p}\sigma\} = 2 (1-p) \bar{F}(k\sqrt{1-p}\sigma).
\end{equation}
For $Y_p$ to have more outliers than $X$ it is sufficient that
\begin{equation}
\label{eq4}
(1-p) \bar{F}(k\sqrt{1-p}\sigma) > \bar{F}(k\sigma).
\end{equation}
There are many cases in which inequality (\ref{eq4}) is true for sufficiently large values of $k$. Let us mention two of them.
\begin{enumerate}
\item Random variable $X$ has exponential tail. More precisely, 
\[ \bar{F}(x) \sim C e^{-a x}, \;\text{as}\; x \to \infty , \]
for some positive constants $C$ and $a$. In this case, inequality (\ref{eq4}) is equivalent for sufficiently large $k$ to
\[(1-p) > Exp\{-a\cdot k \cdot \sigma \cdot (1-\sqrt{1-p})\},\]
which is obviously true for large $k$.
\item $F$ has power tail, that is $\bar{F}(x) \sim C/x^{\alpha}$, where $\alpha>2$ in view of existence of finite second moment. Simple calculations show that (\ref{eq4}) is equivalent as $k \to \infty $ to
\[ (1-p)^{1-\alpha/2} <1. \]
The last inequality is true for $\alpha >2$. 
\end{enumerate}
Let us note that the function $F_p$ has a jump at zero. However, one can obtain similar effect without such jump by using a smoothing procedure, that is by approximating $F_p$ by smooth functions.

"Put tail down" procedure allows us to obtain more outliers in view 
of two its elements. First element consists in changing the tail by smaller, but proportional to previous with coefficient $1-p$. The second element consist in moving a part of mass into origin (or into a small neighborhood of it), which reduces the variance.

\section{Limit Theorems for sums of a random number of random variables}\label{s5}
\setcounter{equation}{0}

One of the arguments used to support the use of stable distributions in Finance is that the observations may be considered as sums of a large number of random variables. However, to have the convergence to stable distribution, the summands themselves must have heavy tails, which seems to be unnatural. Moreover, as we saw in Sections \ref{s2}, \ref{s3}, the probability to observe large number of outliers is not in agreement with real data. However, alternatives to Gaussian distribution are not restricted by stable distributions only. There are limit laws for the sums of i.i.d. random variables for {\it random number of summands}. The theory of such limit distributions has been developed by Robbins, Dobrushin, Gnedenko and others. The description of this theory can be found in the book \cite{KKR}. For some new results, see \cite{KKRT}. As shown in these publications, the sums of a random number of i.i.d. random variables converges to so-called $\nu$-normal distribution. The form of this distribution depends on the law for number of summands. In the case of geometric distribution for the number of summands the limit law is Laplace distribution instead of Gaussian. For the case of transformed negative binomial distribution for the number of summands the role of Gaussian law is played by symmetric gamma distribution. Many other examples are given in \cite{KKRT}. All such distributions have finite second moment. For many of them, the probability (\ref{eq2}) is essentially higher than for Gaussian law. For example, Laplace distribution gives the limit value of probability $p_n$ for $k=3$ equal to $0.0143696$, versus $0.0026998$ for Gaussian distribution. Our suggestion is that such distributions are good alternatives to Gaussian and stable laws.


\begin{thebibliography}{99}

\bibitem{EK}
Ernst Eberlein and Ulrich Keller (1995).
\newblock Hyperbolic distributions in Finance,
\newblock Bernoulli 1 (3), 281-299.

\bibitem{Efr}
Efron, B. (1969). 
 \newblock Student t -test under symmetry conditions. 
 \newblock J. Amer. Statist. Assoc. 64, 1278�1302.
 
\bibitem{IL}
I. A. Ibragimov and Yu. V. Linnik (1971)
\newblock Independent and Stationary Sequences of Random Variables.
\newblock Wolters--Noordhoff Publishing Groningen The Netherlands
(English Translation) 

\bibitem{KKRT}
L.B. Klebanov, A.V. Kakosyan, S.T. Rachev, G. Temnov (2012)
\newblock On a Class of Distributions Stable under Random Summation.
\newblock J. Appl. Prob. 49, 303-318.

\bibitem{KKR}
L. Klebanov, T.J. Kozubowski, S.T. Rachev (2006)
\newblock Ill-Posed Problems in Probability and Stability of Random Sums.
\newblock Nova Science Publishers, New York.

\bibitem{Kl}
Lev B Klebanov (2016)
\newblock No stable distributions in Finance, please.
\newblock arXiv: 1601.00566v2, 1-9.

\bibitem{KV}
Lev B Klebanov and Irina Volchenkova (2015)
\newblock Heavy Tailed Distributions in Finance: Reality or Myth? Amateurs Viewpoint.
\newblock ArXiv: 1507.07735 v.1, 1-17.


\bibitem{LMRSh}
Logan, B. F., Mallows, C. L., Reeds, S. O. and Shepp, L. A. (1973). 
\newblock Limit distribution of self-normalized sums. 
\newblock Ann. Probab. 1 788�809.

\bibitem{KP}
T. J. Kozubowski and A. K. Panorska (1999).
\newblock Multivariate Geometric Stable Distributions in Financial Applications,
\newblock Mathematical and Computer Modeling 29, 83-92 

\bibitem{BM}
 Benoit Mandelbrot (1963)
 \newblock The Variation of Certain Speculative Prices,
 \newblock The Journal of Business, Vol. 36, No. 4, pp. 394-419
 
\bibitem{HB}
S.T Rachev (Editor) (2003).
\newblock Handbook of Heavy Tailed Distributions in Finance, Volume 1: Handbooks in Finance, 
\newblock Book 1 1st Edition, Elsevier B.V. 

\bibitem{RM}
Svetlozar T. Rachev, Stefan Mittnik (2000). 
\newblock Stable Paretian Models in Finance, Wiley

\end{thebibliography}
\end{document}